\documentclass[11pt, a4paper]{article}
\usepackage[T1]{fontenc}
\usepackage{lmodern}
\usepackage{microtype}
\usepackage{fullpage}
\usepackage{authblk}
\usepackage{amsthm}
\usepackage{amssymb}
\usepackage[english]{babel}
\usepackage[utf8]{inputenc}
\usepackage{amsmath}
\usepackage{amsfonts}
\usepackage{graphicx}
\usepackage[colorinlistoftodos]{todonotes}
\usepackage{mathtools}
\usepackage{enumerate}
\usepackage{lineno}
\usepackage{thmtools, thm-restate}
\usepackage{cancel}
\usepackage{xspace}
\usepackage{booktabs}
\usepackage{makecell}
\usepackage[margin=1cm,small]{caption}

\usepackage{color}
\definecolor{darkgreen}{rgb}{0,0.5,0}
\definecolor{darkblue}{rgb}{0,0,0.8}
\definecolor{darkred}{rgb}{0.8,0,0}

\usepackage{tikz}
\usetikzlibrary{shapes,arrows,positioning}

\usepackage{hyperref}
\hypersetup{
   unicode=false,          % non-Latin characters in Acrobat’s bookmarks
   colorlinks=true,        % false: boxed links; true: colored links
   linkcolor=darkred,          % color of internal links (change box color with linkbordercolor)
   citecolor=darkblue,        % color of links to bibliography
   filecolor=magenta,      % color of file links
   urlcolor=black           % color of external links
}

\newcommand\totext[1]{\stackrel{\mathclap{\normalfont{\mbox{\text{\tiny{#1}}}}}}{\to}}

\newtheorem{definition}{Definition}[section]
\newtheorem{lemma}[definition]{Lemma}
\newtheorem{theorem}[definition]{Theorem}

\newcommand{\polylog}{{\rm poly}\log}
\newcommand{\Srole}{\textsf{role}\xspace}
\newcommand{\Scoin}{\textsf{C}\xspace}

\newcommand{\Sslow}{\textsf{I}\xspace}
\newcommand{\Sleader}{\textsf{L}\xspace}

\newcommand{\Sclockmode}{\textsf{timemode}\xspace}
\newcommand{\Sinitiator}{\textsf{injunta}}
\newcommand{\Sfollower}{\textsf{follower}}

\newcommand{\Slevel}{\textsf{level}\xspace}
\newcommand{\Scounter}{\textsf{cnt}\xspace}
\newcommand{\Sslowness}{\textsf{drag}\xspace}

\newcommand{\Stime}{\textsf{phase}\xspace}

\newcommand{\Sslowmode}{\textsf{mode}\xspace}
\newcommand{\Scoinmode}{\textsf{mode}\xspace}
\newcommand{\Sadvancing}{\textsf{adv}\xspace}
\newcommand{\Sstopped}{\textsf{stop}\xspace}

\newcommand{\Sleadermode}{\textsf{leadermode}\xspace}
\newcommand{\Salive}{\textsf{A}\xspace}
\newcommand{\Spassive}{\textsf{P}\xspace}
\newcommand{\Sdead}{\textsf{W}\xspace}

\newcommand{\Strue}{\textsf{true}\xspace}
\newcommand{\Sfalse}{\textsf{false}\xspace}

\newcommand{\Shigh}{\textsf{high}\xspace}
\newcommand{\Slow}{\textsf{low}\xspace}

\newcommand{\Sactive}{\textsf{elev}\xspace}

\newcommand{\Sflipresult}{\textsf{flip}\xspace}
\newcommand{\Sheads}{\textsf{heads}\xspace}
\newcommand{\Stails}{\textsf{tails}\xspace}
\newcommand{\Sheadspresent}{\textsf{void}\xspace}
\newcommand{\Snone}{\textsf{none}\xspace}

\graphicspath{{./figures/}}

\newcommand{\bigo}{\mathcal{O}}
\newcommand{\ignore}[1]{}

\newcommand{\tuple}[1]{\big\langle #1 \big\rangle}

\title{\bf Almost logarithmic expected-time space optimal leader election in population protocols}

\author[1]{Leszek G\k asieniec}
\affil[1]{University of Liverpool, UK}
\author[2]{Grzegorz Stachowiak}
\affil[2]{University of Wroc\l{}aw, Poland}
\author[3]{Przemys\l{}aw~Uzna\'nski}
\affil[3]{
ETH Z\"urich, Switzerland}
\date{}
\begin{document}
\maketitle

\thispagestyle{empty}

%\linenumbers
\begin{abstract}
The model of population protocols refers to a large collection  
of simple indistinguishable entities, frequently called {\em agents}. 
The agents communicate and perform computation through pairwise interactions. 
We study fast and space efficient leader election in population of cardinality $n$ 
governed by a random scheduler, where during each time step the scheduler 
uniformly at random selects for interaction exactly one pair of agents.

We propose the first $o(\log^2 n)$-time leader election protocol.
Our solution operates in expected parallel time $\bigo(\log n\log\log n)$
which is equivalent to $\bigo(n \log n\log\log n)$ pairwise interactions. 
This is the fastest currently known leader election algorithm in which each agent 
utilises asymptotically optimal number of $\bigo(\log\log n)$ states. 
The new protocol 
incorporates and amalgamates successfully 
the power of assorted {\em synthetic coins} with 
variable rate {\em phase clocks}. 
\end{abstract}
%\vfill
%\noindent \textbf{Regular PODC submission. Not eligible for best student paper.}\\ Contact author: {Leszek G\k asieniec},  \url{l.a.gasieniec@liverpool.ac.uk}\\
\setcounter{page}{0}
\clearpage

\section{Introduction}

The computational model of {\em population protocols} 
was introduced
in the seminal paper by Angluin {\em et al.} \cite{DBLP:conf/podc/AngluinADFP04}. 
Their model provides a universal platform for the formal analysis of pairwise interactions 
within a large collection of indistinguishable entities, frequently referred to as {\em agents}.
In this model the agents rely on very limited communication and computation power. 
The actions of agents are prompted by their pairwise interactions with 
the outcome determined by a finite state machine 
$\cal{F}.$
When two agents engage in an interaction they mutually examine the content 
of their local states, and on the conclusion of this encounter
their states change according to the transition function
forming an integral part of $\cal{F}.$
A population protocol terminates with success when eventually all agents stabilise w.r.t. the output (which depends only on their states).

The number of states utilised by the finite state machine $\cal{F}$ constitutes 
the {\em space complexity} of the protocol. 
In the {\em probabilistic variant} of population protocols, introduced in~\cite{DBLP:conf/podc/AngluinADFP04} and
used in this paper, in each step the interacting pair of agents is chosen uniformly at random by the {\em random scheduler}. 
In this variant one is also interested in the {\em time complexity}, i.e., 
the time needed to stabilise (converge) the protocol. 
More recently the studies
on population protocols focus on performance in terms of {\em parallel time} 
defined as the total number
of pairwise interactions (leading to stabilisation) divided 
by the size of the population.
The parallel time can be also interpreted as the local time observed by agents
proportional to the number of interactions it participates in.

Populations protocols attracted studies on several central problems in distributed computing.
This includes the {\em majority} problem, a special instance of {\em consensus}~\cite{DBLP:conf/fct/Fischer83}, 
where the final configuration of states must indicate 
the larger fraction of the population.
The first attempt to computing majority with population protocols can be found  in~\cite{DBLP:conf/podc/AngluinADFP04}.
Later, a neat 3-state one-way protocol for approximate majority was given in~\cite{DBLP:journals/dc/AngluinAE08}. 
In more recent work~\cite{DBLP:conf/podc/AlistarhGV15} Alistarh {\em et al.} consider time-precision 
trade-offs in exact majority population protocols.
Further studies on time-space trade-offs can be found
in~\cite{DBLP:conf/soda/AlistarhAEGR17,DBLP:conf/podc/BilkeCER17} and~\cite{DBLP:conf/soda/AlistarhAG18}, 
where in the latter an asymptotically space-optimal protocol is given. 
The convergence (stabilisation) of majority protocols was also studied in more specific network
topologies~\cite{DBLP:journals/siamco/DraiefV12,DBLP:conf/sss/GasieniecHMS15,DBLP:conf/icalp/MertziosNRS14}, 
as well as in the deterministic setting~\cite{DBLP:conf/opodis/GasieniecHMSS16,DBLP:conf/icalp/MertziosNRS14}.
A useful survey~\cite{DBLP:series/synthesis/2011Michail} discusses 
a range of combinatorial problems suitable for population protocols.

In this paper we study leader election problem where in 
the final configuration a unique agent must converge to a {\em leader state} 
and every other agent has to stabilise in a {\em follower} state. 
While the problem is quite well understood and represented in the literature
 only recently it received greater attention 
in the context of population protocols, partly due to several developments in a related 
model of chemical reactions~\cite{DBLP:conf/wdag/ChenCDS14,DBLP:conf/soda/Doty14}.
In particular, in the follow-up work of Doty and Soloveichik \cite{DBLP:conf/wdag/DotyS15} we learn that leader election 
cannot be solved in sublinear time when agents are equipped with a fixed (constant) number of states. 
On the other hand Alistarh and Gelashvili \cite{DBLP:conf/icalp/AlistarhG15} 
proposed an alternative leader election protocol operating in time $\bigo(\log^3 n)$ and utilising $\bigo(\log^3 n)$ states.
In more recent work~ \cite{DBLP:conf/soda/AlistarhAEGR17} Alistarh {\em et al.} consider a 
trade-off between the number of states utilised by agents and the time complexity of the solution. 
They provide a separation argument distinguishing between slowly stabilising 
population protocols which utilise $o(\log\log n)$ states and 
rapidly stabilising protocols requiring $\Omega(\log \log n)$ states.
This result coincides nicely with another fundamental observation
due to Chatzigiannakis et al.~\cite{DBLP:journals/tcs/ChatzigiannakisMNPS11}
which shows that population protocols utilising $o(\log\log n)$ states can only cope with semi-linear
predicates while presence of $\bigo(\log n)$ states enables computation of symmetric predicates.
Another recent development includes a protocol which elects the leader in time $\bigo(\log^2 n)$ whp 
and in expectation utilising $\bigo(\log^2 n)$ states per agent~\cite{DBLP:conf/podc/BilkeCER17}. 
The number of states was later reduced to $\bigo(\log n)$ 
by Alistarh, Aspnes, Gelashvili in \cite{DBLP:conf/soda/AlistarhAG18} and 
by Berenbrink {\em et al.} 
in \cite{DBLP:conf/soda/BerenbrinkKKO18} through the application of two types of synthetic coins. Please refer to Table~\ref{tab:results} for the summary of past results.

The recent progress in leader election is also aligned with an improved understanding of 
{\em phase clocks} capable of counting parallel time approximately.
The relevant work includes {\em leader-less} phase clocks proposed by Alistarh {\em et al.} \cite{DBLP:conf/soda/AlistarhAG18} and
{\em junta-driven} phase clocks utilised in the fastest currently known $\bigo(\log^2 n)$-time space-optimal leader election algorithm 
\cite{DBLP:conf/soda/GasieniecS18}. 
The concept of phase clocks is also closely related to {\em oscillators} 
used to model behaviour of periodic dynamic systems.
In \cite{DBLP:conf/icalp/CzyzowiczGKKSU15} Czyzowicz {\em et al.} provide 
a thorough study of 3-state oscillators in Lotka-Volterra type population, 
and in the follow up work \cite{DBLP:journals/corr/DudekK17} Dudek and Kosowski
consider information dissemination with authoritative sources.

\paragraph{Our results:}
In this paper we propose the first leader election protocol which stabilises in time $o(\log^2 n)$ assuming the asymptotically 
optimal number of states utilised at each agent.
More precisely, we propose a new $\bigo(\log n\log\log n)$-time protocol
in which each agent operates on $\bigo(\log\log n)$ states.
The solution is always correct but
the improved performance refers to the expected time, i.e.,
the high probability is guaranteed only in time $\bigo(\log^2 n)$
as in the recent work~\cite{DBLP:conf/soda/GasieniecS18}. 
We would like to emphasise that (to the best of our knowledge) the new algorithm is the first space-efficient population protocol which breaks the $\bigo(\log^2 n)$ parallel time barrier in leader election. In addition, the algorithm is almost optimal as any such protocol requires  $\Omega(\log n)$ time to conclude the task. Independently, obtaining a significantly better than $\bigo(\log^2 n)$ time guarantee whp is rather unlikely based on the existing methods (in particular, any protocol synchronised by phase clocks), as even distinguishing between two agents requires a bare minimum of $\Omega(\log n)$ coin-flips with $O(\log n)$-time overhead imposed by phase clocks.
Please note that 
our algorithm similarly to other protocols using non-constant number of states is {\em non-uniform}, i.e., it requires some rough knowledge of $n,$ e.g., to set the size of the phase clock.

\paragraph{Methods used:}
The new leader election algorithm utilises partition of all agents into three sub-populations including {\em coins} ($\Scoin$) responsible for generation of asymmetric coins with $\log\log n$ 
bias levels, 
{\em leaders} ($\Sleader$) among which the unique leader is eventually drawn, and 
{\em inhibitors} ($\Sslow$) designated to maintain variable-rate phase clocks.
A division into sub-populations to reduce space usage at and to accomodate for different roles of agents was previously used in \cite{DBLP:conf/soda/AlistarhAG18,DBLP:conf/podc/GhaffariP16a}. 
The actions of agents in our new protocol are synchronised by phase clocks. 
The first application of phase clocks (goverened by a unique leader) in population protocols refers to \cite{DBLP:journals/dc/AngluinAE08a}. More recently, a novel concept of leader-less phase clocks was introduced in \cite{DBLP:conf/soda/AlistarhAG18}, 
and two-level phase clocks driven by junta 
(group of agents) were utilised successfully in \cite{DBLP:conf/soda/GasieniecS18}. 
The latter synchronisation 
mechanism is also adopted in this paper. 
We also use here synthetic coin-flips as the main symmetry breaking mechanism, where the extraction source of random bits refers to the scheduler. 
In \cite{DBLP:conf/soda/AlistarhAEGR17} the authors explored protocols based on 
uniform coins while in \cite{DBLP:conf/soda/BerenbrinkKKO18} one can find the first use of coins with non-constant bias. This paper 
proposes the first approach in which a larger spectrum of ($\log \log n$) coins, varying in degree of asymmetry and characterised by $\mathrm{poly}(n)$ bias, is utilised. 

Please note that while the fast leader election algorithm presented here builds upon some ideas from \cite{DBLP:conf/soda/GasieniecS18} the significant time improvement is feasible due to several new developments which come in different ``styles and flavours''. 
Firstly, we need biased coins which allow us to elect in time $\bigo(\log n \log \log n)$ a small (logarithmic size) family of active leader candidates. 
Also, during the execution of the protocol we handle low probability ``out-of-sync'' errors differently by creating  a family of $\bigo(\log \log n)$ independent guarantees which certify that at least one active leader candidate remains.  
We also need to guarantee unique leader election during the final reduction stage. 
And while the adopted coin-flipping mechanism enables the relevant reduction in expected time $\bigo(\log n \log \log n)$, we also have to certify the correctness of the reduction process using extra $\bigo(\log \log n)$ states. 
We achieve this goal by creating a family of $\log \log n$ consecutive signals which are used as certification points during the execution of the protocol. The expected delay between different signals is growing exponentially. 
We believe this construction is new in the context of population protocols.

\renewcommand{\arraystretch}{1.3}

\begin{table}[t!]
\centering
{
\small
\begin{tabular}{llll}
\toprule%\noalign{\smallskip}
Paper & States & Time & Runtime guarantee \\
%\noalign{\smallskip}
\midrule
%\noalign{\smallskip}
\cite{DBLP:conf/icalp/AlistarhG15} & $\bigo(\log^3 n)$ & \makecell[l]{$\bigo(\log^3 n)$\\$\bigo(\log^4 n)$} & \makecell[l]{expected\\w.h.p.}\\
\cite{DBLP:conf/soda/AlistarhAEGR17} & $\bigo(\log^2 n)$ & \makecell[l]{$\bigo(\log^{5.3} n\cdot \log \log n)$\\ $\bigo(\log^{6.3} n)$} & \makecell[l]{expected\\w.h.p.}\\
\cite{DBLP:conf/podc/BilkeCER17} & $\bigo(\log^2 n)$ & $\bigo(\log^2 n)$ & w.h.p.\\
\cite{DBLP:conf/soda/AlistarhAG18} & $\bigo(\log n)$ & $\bigo(\log^2 n)$ & expected\\
\cite{DBLP:conf/soda/BerenbrinkKKO18} & $\bigo(\log n)$ & $\bigo(\log^2 n)$ & w.h.p.\\
\cite{DBLP:conf/soda/GasieniecS18} & $\bigo(\log \log n)$ & $\bigo(\log^2 n)$ & w.h.p.\\

This work & $\bigo(\log \log n)$ & $\bigo(\log n\cdot \log \log n)$ & expected\\
\bottomrule
\end{tabular}
}
\caption{Recent progress in leader election via population protocols.}
\label{tab:results}
\end{table}

\paragraph{Related work:}
Leader election is one of the fundamental problems in Distributed Computing besides broadcasting, 
mutual-exclusion, consensus, see, e.g., an excellent text book by Attiya and Welch \cite{Attiya:2004:DCF:983102}. 
The problem was originally studied in networks with nodes having distinct 
labels \cite{DBLP:conf/ifip/Lann77}, where an early work focuses on the ring topology in synchronous 
\cite{DBLP:journals/jacm/FredericksonL87,DBLP:journals/cacm/HirschbergS80} 
as well as in asynchronous models \cite{Burns-TechRep,Peterson:1982:ONL:69622.357194}. 
Also, in networks populated by mobile agents the leader election was studied 
first in networks with labelled nodes \cite{DBLP:conf/aiccsa/HaddarKMMJ08}. 
However, very often leader election 
is used as a powerful symmetry breaking mechanism enabling feasibility and
coordination of more complex protocols in systems based on uniform (indistinguishable) entities. 
There is a large volume of work 
\cite{DBLP:conf/stoc/Angluin80,DBLP:journals/jal/AttiyaS91,DBLP:journals/jacm/AttiyaSW88,DBLP:conf/istcs/BoldiSVCGS96,DBLP:conf/podc/BoldiV99,DBLP:conf/wdag/YamashitaK89,DBLP:journals/tpds/YamashitaK96} 
on leader election in anonymous networks. 
In \cite{DBLP:conf/wdag/YamashitaK89,DBLP:journals/tpds/YamashitaK96} we find a good characterisation of message-passing networks in 
which leader election is feasible when the nodes are anonymous. 
In \cite{DBLP:conf/wdag/YamashitaK89}, the authors study the problem of leader election in general networks 
under the assumption that node labels are not unique. 
In \cite{DBLP:journals/jpdc/FlocchiniKKLS04}, the authors study feasibility and message complexity of leader election in rings with possibly
non-unique labels, while in \cite{DBLP:journals/fuin/DobrevP04} the authors provide solutions to a generalised
leader election problem in rings with arbitrary labels. 
The work in \cite{DBLP:journals/dc/FuscoP11} focuses on the time complexity of leader election in
anonymous networks where this complexity is expressed in terms of multiple network parameters. 
In \cite{DBLP:journals/dc/DereniowskiP14}, the authors study feasibility of leader election for anonymous agents 
that navigate in a network asynchronously. Another important study on 
trade-offs between the time complexity and knowledge available in anonymous 
trees can be found in recent work of Glacet {\em et al.} \cite{DBLP:conf/soda/GlacetMP16}.
Finally, a good example of recent extensive studies on the exact space complexity 
in related models refers to plurality consensus. In particular, in \cite{DBLP:conf/icalp/BerenbrinkFGK16} 
Berenbrink {\em et al.} proposed a plurality consensus protocol for $C$ original opinions 
converging in $\bigo(\log C \log\log n)$ synchronous rounds using only $\log C + \bigo(\log\log C)$ 
bits of local memory. They also show a slightly slower solution converging in 
$\bigo(\log n\log\log n)$ rounds utilising only $\log C + 4$ bits of the local memory. 
They also pointed out that any
protocol with local memory $\log C + \bigo(1)$ has the worst-case running time $\Omega(k).$
In \cite{DBLP:conf/podc/GhaffariP16a} Ghaffari and Parter propose an alternative algorithm converging in time
$\bigo(\log C\log n)$ in which all messages and the local memory are bounded to 
$\log C + \bigo(1)$ bits. 
Some work on utilisation of random walk in plurality 
consensus can be found in \cite{DBLP:conf/soda/BecchettiCNPS15,DBLP:conf/sss/GasieniecHMS15}.

\section{Preliminaries}
We study here 
protocols defined on populations of identical agents in which 
a dedicated {\em random scheduler} connects 
(sequentially or in parallel) 
agents in pairs {\em uniformly at random}.
We assume that all $n$ agents start in the same initial state.
We adopt the classical model of population protocols~\cite{DBLP:conf/podc/AngluinADFP04,DBLP:journals/dc/AngluinAE08} 
in which each interaction refers to an ordered pair of agents
({\textsf responder}, {\textsf initiator}). 
Each interaction triggers an update of states in both agents according to some predefined 
deterministic {\em transition function}, where the update of relevant states is denoted by $A + B \to C + D$.
We focus on two complexity measures including {\em space complexity} defined as the {\em number of states} 
utilised by each agent, and {\em time complexity} reflecting the total number of interactions required to stabilise the population protocol.
We also consider {\em parallel time} defined as the total number of interactions divided by the size of the population.
This time measure can be also seen as the local time observed by an agent, i.e.,
the number of pairwise interactions in which the agent is involved in.  
We aim at protocols formed of $\bigo(n\cdot \polylog n)$ interactions 
equivalent to the parallel running time $\bigo(\polylog n).$

In order to maintain clarity of presentation each state has
a name drawn from either a fixed size set of suitable names 
or a small range of integer values. 
However, when it is clear from the context we tend to omit the name of this field. 
Moreover, since each node belongs to exactly one of 3 sub-populations, 
for simplicity we shorten the notation omitting the part $\Srole=$ and writing for example $\Scoin\tuple{\dots}$ instead of $\tuple{\Srole=\Scoin, \dots}$.
This notation allows us to refer only to the relevant fields, i.e., those affected during one particular type of interaction. One should keep in mind also that interactions may trigger several non-conflicting rules. 
For example, rules of transition of clocks happen in parallel to the rules of transition of coins.

Consider an event $X,$ and let $\eta>0$ be some predefined constant.
We say that an event occurs with {\em negligible} probability,
if there is an integer $n_0,$ s.t., the probability of this event
for $n>n_0$ is at most $n^{-\eta}$.
An event occurs {\em with high probability} (whp) if its probability
is at least $1-n^{-\eta}$ for $n>n_0$
If the event refers to a behaviour of an algorithm, we say it occurs with high probability if 
the constants used in the algorithm can be fine-tuned so that the probability of this event is at least $1-n^{-\eta}$. 
Analogously an event $X$ occurs {\em with very high probability} (wvhp)
if for any $a>0$ there exists an integer $n_a$
such that event $X$ occurs with probability at least $1-n^{-a}$ when $n>n_a$.
In particular, if an event occurs with probability $1-n^{-\omega(1)}$, it occurs with very high probability.

\section{Phase clock}
The actions of our leader election protocol are coordinated by a phase clock utilising junta of clock leaders.
A similar approach can be found in~\cite{DBLP:conf/soda/GasieniecS18}.
The junta leaders are drawn from sub-population {\em coins} denoted by $\Scoin$.
With the help of the phase clock every agent in $\Scoin$ keeps track of 
$\Stime \in \{0,1,\ldots,\Gamma-1\},$ 
for a suitable large constant $\Gamma$, 
and maintains its $\Sclockmode \in \{\Sinitiator,\Sfollower\}$.  Let $+_{\Gamma}$ denote addition modulo $\Gamma$ and 
$$\max\vphantom{0}_{\Gamma}(x,y) = \begin{cases} \max(x,y)\quad \text{\!if } |x-y| \le \Gamma/2,\\ \min(x,y)\quad\text{if } |x-y| > \Gamma/2.\end{cases}$$
The transition rules of interaction with respect to the phase clock include:
$$
\tuple{\Sfollower, \Stime = t_1} + \tuple{\Stime = t_2} 
\to
\tuple{\Sfollower, \Stime = T_1} + \tuple{\Stime = t_2}
$$
$$
\tuple{\Sinitiator, \Stime = t_1} + \tuple{\Stime = t_2} 
\to
\tuple{\Sinitiator, \Stime = T_2} + \tuple{\Stime = t_2},
$$
where $T_1 = \max\vphantom{0}_{\Gamma}(t_1,t_2),$  $T_2 = \max\vphantom{0}_{\Gamma}(t_1,t_2+_{\Gamma}1)$.
Agents are initialised to $\tuple{\Sfollower, \Stime = 0}$. 
During execution of coin preprocessing protocol, see Section~\ref{s:Coins}, 
some agents in $\Scoin$ become \emph{junta members}. 
We say that the phase clock {\em passes through 0} 
whenever its current phase $x$ of the clock is reduced in absolute terms.
We denote this transition by $\ \totext{0}.$ 

\begin{definition}[c.f., \cite{DBLP:conf/soda/GasieniecS18}]
Passes through 0 of agents $a$ and $b$ are \emph{equivalent} if they both occur in a period 
when the respective agent's clock phases $x_a$ and $x_b$ satisfy $3\Gamma/4 <_{\Gamma} x_a,x_b <_{\Gamma} \Gamma/4$.
\end{definition}
\begin{theorem}[c.f., Theorem 3.1 and Fact 3.1 in \cite{DBLP:conf/soda/GasieniecS18}]
\label{th:phaseclock}
For any constant $\varepsilon,\eta,d>0$, there exists a constant $\Gamma,$ s.t., 
if the number of junta members is at most $n^{1-\varepsilon}$ at any time whp $1-n^{-\eta},$ 
the following conditions hold whp
until each agent completes $n^{\eta}$ passes through 0:
\begin{itemize}
\item All passes through 0 form equivalence classes for
all agents and the number of interactions between
the closest passes through 0 in different equivalence
classes is  at least $d\cdot n \log n.$
\item The number of interactions between two subsequent
passes through $0$ in any agent is $\bigo(n \log n)$.
\end{itemize}
\end{theorem}

A period between two agent's passes through zero is called a {\em round}.
By Theorem \ref{th:phaseclock},
the rounds for different agents form equivalence classes whp\ that are referred to as {\em rounds of the protocol}.
The {\em updated} agent is the one which acts 
as responder during the relevant interaction. 
Interactions with both start- and end-phase in $\{0,1,\dots,\Gamma/2-1\}$ are denoted by $\ \totext{early}\ $ 
and those with start- and end-phase in $\{\Gamma/2, \dots, \Gamma-1\}$ are denoted by $\ \totext{late}\ $. 
Finally, applying Theorem~\ref{th:phaseclock} and with $\Gamma$ being twice as big as required by Theorem~\ref{th:phaseclock}, 
we can guarantee that passes through $0$ and through $\Gamma/2$ form strictly separate equivalence classes.

\section{High level description}
%\todo[inline]{przerzucić 'jak działa' do rozdziałów, a zostawić strukturę i interfejsy}
An execution of our algorithm consists of three consecutive epochs whp. 
These include the \emph{initialisation} epoch, the \emph{fast elimination} epoch 
and the \emph{final elimination} epoch.
For the case when any epoch fails, which happens with negligible probability, 
we use as backup the slow leader election protocol working in time $\bigo(n\log n)$~\cite{DBLP:journals/dc/AngluinAE08}
During the initialisation epoch the whole population is divided into sub-populations,
where the descriptor $\Srole \in \{\Scoin, \Sslow, \Sleader\}$ 
differentiates agents between the three sub-populations of \emph{coins}, 
\emph{inhibitors} and \emph{leaders} respectively.
At the start of the protocol all agents are subject to symmetry breaking rules.
Each agent gets assigned to one of the three roles (or gets deactivated), and this role is never changed. 
The two symmetry breaking rules adopted during the initial partition process are as follows:
\begin{align}
\label{eq:0_elimination}
0 + 0 &\to \text{X} + \Sleader, & \text{X}+\text{X} &\to \Scoin+\Sslow,
\end{align}
where $0$ describes an agent before initialisations, 
and $X$ refers to an intermediate stage before entering sub-population $\Scoin$ or $\Sslow.$ 

During the initialisation epoch a {\em junta} of size at most $n^{0.77}$ is elected from $\Scoin$ whp, 
which allows to start the phase clock using this junta as clock leaders.
This phase clock synchronises all actions of our algorithm until it concludes.
In our approach it is important to terminate the initialisation epoch and 
in turn to stabilise the roles of agents in time $\bigo(\log n)$.
With this in mind we adopt two extra rules, s.t., 
whenever a node in state $0$ or $\text{X}$ reaches the end of the first round, it deactivates itself:
\begin{align}
\label{eq:X_deactivation}
0 + \star\ &\totext{0} \mathsf{D} + \star, & X + \star\ &\totext{0} \mathsf{D} + \star,
\end{align}
where $\mathsf{D}$ denotes {\em deactivated} agents that, except for passing clock state, do not play 
any meaningful role in the leader election protocol.

\begin{lemma}
\label{lem:noninitialized}
With high probability, only $\bigo(n/\log n)$ agents are not initialised in the course of the protocol, i.e., 
$n - \bigo(n/\log n)$ agents join \Scoin, \Sslow or \Sleader during the first $\bigo(n \log n)$ interactions.
\end{lemma}
\begin{proof}
By Theorem~\ref{th:phaseclock}, the first round of the phase clock is completed with 
high probability during the first $d \cdot n \log n$ interactions, for some constant $d$.
We first show that after $4\cdot n\log n$ interactions at most $n/\log n$ not yet initialised in state $0$ agents remain.
Let $\mathcal{X}$ be the random variable denoting the number of agents in state $0$.

Assume $\mathcal{X} = \alpha n,$ for some $\alpha>0.$ %and $\mathcal{X} = \Omega(\log^2 n)$.
We prove that the number of interactions it takes to reduce $\mathcal{X}$ by a factor of 2 is 
at most $4n/\alpha$, with very high probability.
Let $\sigma$ be a 0-1 sequence of length $4n/\alpha$ referring to the relevant $4n/\alpha$ interactions.
In this sequence an entry is set to 1 if during the corresponding interaction 
the number of not yet initialised agents is reduced, and 0 otherwise. 
For as long as $\mathcal{X}>\alpha n/2$, the probability of having 1 at each position in $\sigma$ 
is at least $\alpha^2/4$ and in turn
%When $\mathcal{X}\leq \alpha n/2$ entry in chosen by random and probability of drawing 1 is $\alpha^2 n/4$.
the expected number of 1s in $\sigma$ is at least $\alpha n$.
Thus by Chernoff bound the number of 1s in $\sigma$ is at least $\alpha n/2$ with very high probability.
%A sufficient condition for $\mathcal{X}\leq \alpha n/2$ after $2n$ interactions 
%is that $\sigma$ has at least $\alpha n/4$ ones,
This implies that at least $\alpha n/2$ agents in state $0$ get initialised.
And iterating this process $\log\log n$ times we get reduction of agents in state $0$ to $n/\log n$
in at most 
%$4n(1+2+4+\cdots+2^{\log\log n-1}) = 
$\bigo(n\log n)$ consecutive interactions. 
A similar reasoning can be used for agents in the intermediate state $X$ in the next 
(subsequent to reduction of agents in state 0) $\bigo(n\log n)$ interactions. 
Thus one can conclude that after $\bigo(n\log n)$ initial interactions 
the number of not yet initialised agents is at most $2n/\log n$.
\end{proof}

Using Lemma~\ref{lem:noninitialized} one can immediately conclude that during the first round
all $\bigo(n/\log n)$ not initialised agents, i.e., those not given roles $\Scoin,\Sleader$ or $\Sslow$, 
become deactivated with high probability by rule \eqref{eq:X_deactivation}.
Below we explain functionality of the three adopted sub-populations.

\begin{description}
\item[Coin] Agents in this group differentiate themselves into non-empty
levels $0,1,2,\ldots,\Phi$, where $\Phi=\log\log n -3.$
The number of agents on level $\Phi$ is at most $n^{0.77}$
and these agents form the junta running the phase clock.
The levels are also used to simulate $\Phi+1$ types of asymmetric coins, s.t., 
if the probability of drawing heads by $\ell$-th coin is $q$
the probability of drawing heads by $(\ell+1)$-st coin is roughly $q^2$.
In terms of implementation, tossing $\ell$-th asymmetric coin is realised by an agent 
interacting with another agent as a responder.
And the outcome is {\em heads} if the initiator is a coin on level $\ell$ or higher.

\iffalse
$\Phi+1 = \log \log n-2$ distinct \emph{levels}, where level $i$ desired size is $\bigo(n/2^{2^i})$ and the size of levels span from $\bigo(n)$ to $\bigo(n^c)$ for some $0.1 \le c \le 0.9$. The role they perform is two-fold. First, level $\Phi$ of size $n^c$ is used as an initiator for the phase clock, which thus requires only $\bigo(1)$ states to oscillate with $\bigo(\log n)$ time between ticks. Second, level $i$ is used to simulate synthetic biased coin of probability $1/2^{2^i}$, giving us access to wide range of coins of probabilities approximately $1/2, 1/4, \ldots, \ldots 1/\sqrt{n}$.
\fi

\item[Leader] Agents in this group are leader candidates, i.e.,
each agent in this group has a chance to become the unique leader.   
In due course the number of candidates is reduced to one.
The main challenge is in fast but also safe candidate elimination, i.e., 
we need to guarantee that our protocol does not eliminate all candidates.  

\item[Inhibitor] Agents in this group are split into $\Psi = \Theta(\log \log n)$ 
distinct subgroups, with the expected cardinality of each subgroup $i$ at level $\bigo(n/2^i).$
I.e., the sizes of subgroups span from $\bigo(n)$ to $\bigo(n/\log^c n),$ 
for some constant $c>0$. 
We target with this system of subgroups specific points in time, where the $i$th subgroup is responsible for 
$\bigo(2^{\bigo(i)})$ round in expectation. 
This system of groups is used to guide through the final elimination process when
we safely reduce the number of leaders from $\bigo(\log n)$ to a single one.
\end{description}

The first (initialisation) epoch generates at least one leader candidate and
with high probability the number of candidates is almost $n/2$.
The protocol will eventually elect a single leader among all leader candidates 
in $\Sleader$ during the second and the third epoch.
The second epoch related to fast elimination reduces the number of \emph{active}
(not withdrawn yet) leader candidates to $\bigo(\log n)$ agents 
in time $\bigo(\log n\log\log n)$
with high probability.
The fast elimination uses the sub-population $\Scoin$ to simulate assorted biased coins.
The third epoch eliminates all but one competitor 
which becomes the unique leader, and is successful in achieving this goal whp.
This process requires utilisation of inhibitors from $\Sslow$ to
guarantee survival of at least one leader candidate.
The third epoch elects a single leader in $\bigo(\log n\log\log n)$ expected time
and in $\bigo(\log^2 n)$ time with high probability.

Independently, we simultaneously perform actions of the slow constant-space leader election protocol, in which if two leader candidates interact, exactly one of them gets eliminated.
This slow protocol causes a slow depletion of leader candidates, which does not
have a noticeable effect on the second and the third epochs of the protocol.
This depletion assures election of one leader in cases when either the phase clock gets desynchronised or
all leader candidates become passive (marked for elimination) during the last two epochs. 

%Subpopulations of coins ($\Scoin$) and inhibitors ($\Sslow$) 
%help this process.
%Election of a single leader is made by elimination of all other candidates.
%The protocol consists of three phases.
%The first phase 

The leader candidates elimination process during the second 
and the third epoch works as follows.
The protocol operates in consecutive rounds, 
each taking time $\bigo(\log n).$
For each agent a round is defined as the time between 
two subsequent passes of the phase clock through value zero.
In the first half of each round still active leader candidates flip a coin to decide whether they intend to survive ({\em heads}) this round or not ({\em tails}).
If any {\em heads} are drawn during this round, the relevant information 
is distributed (via one-way epidemic~\cite{DBLP:journals/dc/AngluinAE08}) to all agents during the second half of the round.
This results in elimination of all active candidates which drew {\em tails}.
However, if no {\em heads} are drawn the round is considered {\em void}.

In the fast elimination process we utilise asymmetric coins 
implemented through interactions with agents in diverse population of $\Scoin$.
The first asymmetric coin $\Phi$ is used 4 times to reduce the population
of active leader candidates to size at most $n/n^{0.77}$.
Further we use each of asymmetric coins $\Phi-1,\Phi-2,\ldots,1$ exactly twice.
Using a biased coin with {\em heads} coming with probability $q$ guarantees whp 
reduction of active leader candidates by a factor close to $q$.
On the conclusion of this process (all coins are used) 
the number of active leader candidates is down to $\bigo(\log n)$
whp.

In contrast, in the third epoch symmetric almost fair coins are used in
the elimination process indefinitely. 
This results in elimination of all but a single leader candidate
in $\bigo(\log\log n)$ rounds. 
In order to guarantee that the protocol is \emph{always} correct, i.e., 
the last alive leader candidate is never eliminated, we use the support of
agents in sub-population $\Sslow.$

\section{Coins}\label{s:Coins}

Let $\Phi = \lfloor \log \log n \rfloor - 3$.
The states of coins, i.e., agents belonging to sub-population $\Scoin$ 
store the following information: $\Slevel \in \{0,1,2,\ldots,\Phi\},$ 
reflecting the level of asymmetry, and
$\Scoinmode \in \{\Sadvancing,\Sstopped\}$
indicating whether a coin is still willing to increment its level.
We also need an extra constant space to store the current state of the phase clock.
When formed after application of split rule \eqref{eq:0_elimination} 
each coin is initialized to $\Scoin\tuple{\Slevel=0,\Sadvancing}$.

\paragraph{Coin preprocessing:}
In what follows we introduce the rules governing level incrementation. 
Note that these closely resemble the rules from \emph{forming junta} protocol proposed in \cite{DBLP:conf/soda/GasieniecS18}.
\begin{align*}
&\Scoin\tuple{\Slevel=x,\Sadvancing}+ Y
\to
\Scoin\tuple{\Slevel=x,\Sstopped} + Y,&\text{for } Y\neq\Scoin,\\
&\Scoin\tuple{\Slevel=x,\Sadvancing}+ \Scoin\tuple{\Slevel=y}
\to
\Scoin\tuple{\Slevel=x,\Sstopped} + \Scoin\tuple{\Slevel=y},&\text{for } x > y,\\
&\Scoin\tuple{\Slevel=x, \Sadvancing} + \Scoin\tuple{\Slevel=y} 
\to
\Scoin\tuple{\Slevel=x+1, \Sadvancing} + \Scoin\tuple{\Slevel=y},&\text{for } x \le y, x<\Phi.
\end{align*}
Once the $\Slevel$ of a coin in $\Scoin$ reaches $\Phi$ it stops growing.
Moreover, we give name  $\Sinitiator$ to all coins which managed to reach $\Slevel = \Phi.$ 
In order to characterise properties of coins we formulate a series of lemmas.
Let $n_C$ be the total number of coins. 
By Lemma~\ref{lem:noninitialized} and rules~\eqref{eq:0_elimination} and \eqref{eq:X_deactivation}, 
$n_C = \frac{n}{4} - \bigo(n/\log n)$ with very high probability.
Let $C_\ell$ be the number of coins which reach level $\ell$ or higher.
The value of $C_\ell$ depends on the execution thread of the protocol.
We first observe that $n_C=C_0$,
and further estimates on $C_\ell,$ for $\ell>,0$ are determined by 
Lemmas~\ref{FJ:reductB} (upper bound) and \ref{FJ:reductB1} (lower bound).

\begin{lemma}[Lemma 4.2. in \cite{DBLP:conf/soda/GasieniecS18}]
\label{FJ:reductB}
Assume $n^{-1/3}\leq q<1$ and $C_\ell=q\cdot n$, 
then $C_{\ell+1}\leq \frac{11}{10} q^2\cdot n$ with very high probability.% $1-e^{-n^{1/3}/300}$.
\end{lemma}

The lower bound argument (similar to the proof of Lemma~\ref{FJ:reductB}) is given below.

\begin{lemma}\label{FJ:reductB1}
Assume $n^{-1/3}\leq q<1$ and $C_\ell=q\cdot n$, 
then $C_{\ell+1}\geq \frac{9}{20} q^2\cdot n$ wvhp.% $1-e^{-n^{1/3}/800}$.
\end{lemma}
\begin{proof}
	Each coin contributing to value $C_\ell$ arrives at level $\ell$ during some interaction $t$.
	These coins arrive sequentially. 
	Consider $(i+1)$-st coin $v$ that got to level $\ell$.
	At the time the coin arrives there are already $i$ coins on levels $\ell'\geq \ell$.
	Consider the first interaction $\tau$ succeeding $t$ 
	in which coin $v$ acts as the responder.
	During this interaction the initiator is a coin on level $\ell'\geq \ell$ with probability $p_{\tau}\geq i/n.$
	Thus $v$ moves to level $\ell+1$ with probability at least $i/n$ as
	otherwise the responder would end up in state $(\ell,0)$ and would not contribute to $C_{\ell+1}$. 
	Consider now the sequence of $C_\ell$ such interactions $\tau$, in which each of $C_\ell$ coins 
	act as responder after getting to level $\ell$.
	We can attribute to this sequence a binary $0$-$1$ sequence $\sigma$ of length $C_\ell$, s.t.,
	if during interaction $\tau$ a coin ends up in state $(\ell,0)$, the respective entry in $\sigma$ becomes 0, and otherwise this entry becomes 1 (this happens with probability at least $p_{\tau}$).
	The expected number of these $1$s is at least $\sum_i i/n= (C_\ell-1)C_\ell/2n=(q^2\cdot n-q)/2$.
	And by Chernoff bound $C_{\ell+1}<\frac{9}{20}q^2 \cdot n$ with very high probability.% at most $e^{-A^2 n/800}<e^{-n^{1/3}/800}$.
\end{proof}

\begin{lemma}\label{FJ:loglog-4}
For $n$ large enough and $\Phi=\lfloor\log\log n\rfloor -3$ we have
$n^{0.45}\le C_\Phi\le n^{0.77}$ wvhp.
\end{lemma}
\begin{proof}
	We start with $9n/40\leq n_C=C_0\leq n/4$ with very high probability.
	By Lemma \ref{FJ:reductB} and Lemma \ref{FJ:reductB1} iterated $\ell$ times we conclude that with very high probability
	\[
	(9/20)^{2^{\ell+1}-1}\cdot \frac{n}{2^{2^{\ell+1}}}\leq C_{\ell}\leq (11/10)^{2^\ell-1}\cdot\frac{n}{2^{2^{\ell+2}}}.
	\] 
	Note that if we adopt $\Phi=\lfloor\log\log n\rfloor-3$, we get 
	\[
	C_\Phi\geq (9/20)^{2^{\Phi+1}}\cdot \frac{n}{2^{2^{\Phi+1}}}\geq n\cdot (9/40)^{2^{\log\log n-2}}\geq \frac{n}{2^{2.2\cdot 2^{\log\log n}/4}}\geq n/n^{0.55}=n^{0.45}.
	\]
	On the other hand
	\[
	C_\Phi\leq (11/10)^{2^\Phi}\cdot\frac{n}{2^{2^{\Phi+2}}}\leq n\cdot (11/160)^{2^{\log\log n-4}}\leq \frac{n}{2^{3.8\cdot 2^{\log\log n}/16}}\leq n/n^{0.23}= n^{0.77}.\qedhere
	\]
\end{proof}

\begin{lemma}[Analogue of Lemma 4.5. in \cite{DBLP:conf/soda/GasieniecS18}]
\label{lem:coins_speed}
The bounds from Lemma~\ref{FJ:reductB}, Lemma~\ref{FJ:reductB1} and Lemma~\ref{FJ:loglog-4} hold after $\bigo(n \log n)$ interactions.
\end{lemma}
\begin{proof}
{\it (Sketch)} In coin preprocessing protocol we need to stabilise first sub-population 
of coins in time $\bigo(\log n).$ 
The time complexity analysis of the remaining part of the protocol 
is analogous to the one used in forming junta protocol in \cite{DBLP:conf/soda/GasieniecS18}.
\end{proof}

%%%%%%%%%%%%
\begin{figure}[t]
\tikzstyle{line} = [draw, -latex', thick]
%\tikzstyle{empty} = [draw, text, text centered, node distance = 6em]
\tikzstyle{block} = [draw, rectangle, text centered, text width = 6em, minimum width=6em, minimum height=2em]
\tikzstyle{circ} = [draw, circle, text centered, node distance = 6em, text width = 4em]

\centering
\begin{tikzpicture}[scale=0.7, transform shape]

\node[circ, anchor=center] (c0){$C_0 \approx \frac{n}{4}$};
\node[circ, right=6em of c0] (c1){$C_1 \approx \frac{n}{16}$};
\node[text width = 2em, align = center, right=6em of c1] (mid){$\dots$};
\node[circ, right=6em of mid] (cphi1){$C_{\Phi-1} \approx n^{1-a/2}$};
\node[circ, right=6em of cphi1] (cphi){$C_{\Phi} \approx n^{1-a}$};
\node[block, below=2em of cphi] (clock){\makecell[c]{junta\\ (clock)}};
\node[block, above=2em of c0] (flip0){\makecell[c]{coin \#0\\ bias $\frac14$}};
\node[block, above=2em of c1] (flip1){\makecell[c]{coin \#1\\ bias $\frac1{16}$}};
\node[block, above=2em of cphi1] (flipphi1){\makecell[c]{coin \#$(\Phi-1)$\\ bias $n^{-a/2}$}};
\node[block, above=2em of cphi] (flipphi){\makecell[c]{coin \#$\Phi$\\ bias $n^{-a}$}};

\path[line, double] (c0) -- node[anchor=south, name=arr0]{$\frac14$} (c1);
\path[line, double] (c1) -- node[anchor=south, name=arr1]{$\frac1{16}$} (mid);
\path[line, double] (mid) -- node[anchor=south]{$n^{-a/4}$} (cphi1);
\path[line, double] (cphi1) -- node[anchor=south, name=arrphi1]{$ n^{-a/2}$} (cphi);
\path[line, thick, dashed] (cphi) -- (clock);
\path[line, thick, dashed] (c0) -- (flip0);
\path[line, thick, dashed] (flip0) -- (arr0);
\path[line, thick, dashed] (c1) -- (flip1);
\path[line, thick, dashed] (flip1) -- (arr1);
\path[line, thick, dashed] (cphi1) -- (flipphi1);
\path[line, thick, dashed] (flipphi1) -- (arrphi1);
\path[line, thick, dashed] (cphi) -- (flipphi);

\end{tikzpicture}
\caption{An idealized scheme of coin sub-populations and their relation to biased coins. In the picture $0.23\le a\le 0.55$. Solid lines denote evolution of the population and dashed lines refer to the relevant functionality.}
\end{figure}
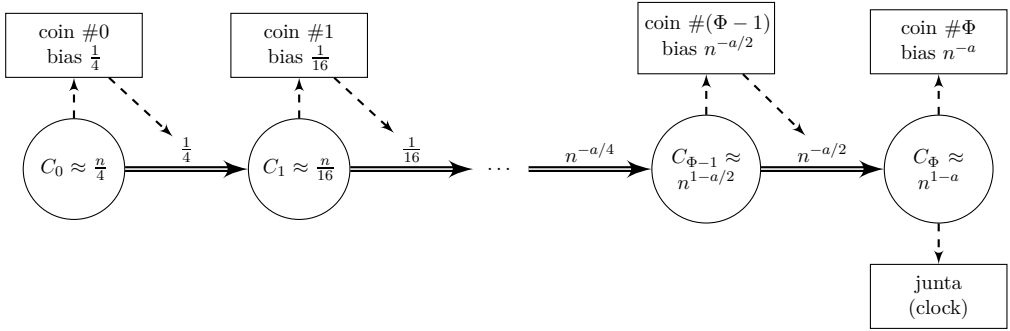

%%%%%%%%%%%%

%\todo[inline]{A short comment providing further detail would be appreciated.}

\section{Fast elimination}
\label{sec:fast}

%The population that will eventually elect a single leader. Let $\Psi = \Theta(\log \log n)$ of value to be fixed later. Each node state is: $\Slevel \in \{0,1, \ldots, \Phi\}$, $\Sslowness \in \{0, 1, \ldots, \Psi\}$,  $\Sleadermode \in \{\Salive,\Starget,\Spassive,\Sdead\}$ (short for \emph{Active}, \emph{Target}, \emph{Passive} and \emph{Withdrawn}), %, $\text{coinflip} \in \{\text{waiting},\text{done}\}$ 
%and $\Sbroadcast \in \{\Strue,\Sfalse\}$. Each node is initialized to $\Sleader\tuple{\Slevel = \Phi, \Sslowness = 0, \Salive, \Sfalse}$. The rules of the protocol make sure that either $\Slevel = 0$ or $\Sslowness = 0$, thus making sure that the number of states used is $\Theta(\log \log n)$.

The goal in fast elimination epoch is to reduce the number of 
active leader candidates to $\bigo(\log n)$ whp.
We also guarantee that at least one agent remains in the group of 
active leaders $\Salive$ whp.
All other leader candidates join group $\Spassive$ of passive agents.  
%in order to make possible electing 
%the leader in an unlikely event when no leader candidate remains active.

%Let $\Phi' = 2\Phi+2$.
The state of each leader candidate in this epoch consists of: 
%the following components: 
$\Scounter \in \{0,1, \ldots, 2\Phi+3\}$,
%$\Scount \in \{1,2,3,4\}$ (counter of iterations for one $\Slevel$), 
$\Sleadermode \in \{\Salive,\Spassive,\Sdead\}$ (in fast elimination  
$\Sdead$ standing for withdrawn is not used),
$\Sflipresult \in \{\Snone,\Sheads,\Stails\}$,
$\Sheadspresent \in \{\Strue,\Sfalse\}$ (telling whether the round is void),
and a constant number of phase clock values.
Each leader candidate is initialised at the beginning of the first round of the second epoch to
$\Sleader\tuple{\Scounter = 2\Phi+3, \Salive, \Snone, \Sheadspresent=\Strue}$.

%We first describe behavior for candidate nodes in the first phase, that is when $\Sslowness = 0$.
After the first round of the phase clock, when the roles of all agents are fixed and levels of all coins are computed whp,
agents enter the fast elimination epoch. This is ensured by starting the counter at one larger than the intended number of coin uses.
At the beginning of the fast elimination all leader candidates are active ($\Salive$).
In fast elimination we use the sub-population $\Scoin$ of coins 
as the source of $\Phi$ different types of asymmetric coins.
The coin result is generated, when a leader candidate interacts 
with another agent acting as the responder.
The outcome of using $\ell$-th biased coin is {\em heads} when 
the interaction refers to a coin on level at least $\ell,$ and {\em tails} otherwise.
When $C_\ell=q\cdot n$, the probability of drawing {\em heads} at this level is $q$.
Thus when there are substantially more than $1/q$ active leader candidates 
almost certainly at least one of them has to draw {\em heads}.
In turn the number of active leader candidates will be reduced 
by factor of $1/q$ in expectation.
%Even if only agents with {\em heads} drawn remain active this results 
%on average decreases the sub-population of active candidates by the factor of $A$.
On the other hand, if the number of active leader candidates 
does not exceed $1/q$, no agent may draw {\em heads}.
In order to have good understanding of the situation 
the agents with {\em heads} drawn inform others (using one-way epidemic) 
about this fact.
Thus if an agent draws $\Stails$ and receives a message about other agent(s) having $\Sheads$, it
can safely become passive ($\Spassive$).
This elimination cycle can be carried in one round in time $\bigo(\log n)$.

During fast elimination active leader candidates utilise coins, s.t.,
each coin $1,2,\ldots\Phi-2,\Phi-1$ is used exactly twice and 
coin $\Phi$ is applied four times.
In other words, the elimination process can be represented by a sequence $(\gamma)_{1}^{2\Phi+2} = [1,1,2,2,\ldots,\Phi-1,\Phi-1,\Phi,\Phi,\Phi,\Phi]$ which tells us which coin $\Slevel$ is used with what $\Scounter$ value.
In total, the elimination process operates in $\bigo(\log\log n)$ rounds 
translating to parallel time $\bigo(\log n\log\log n).$ We are also able to guarantee
reduction of the number of remaining active leader candidates to $\bigo(\log n)$ whp. 

The following transitions are used in the second epoch.
When the phase clock passes through zero we have
\begin{align}
\label{FE:zeroA}
&\Sleader\tuple{\Scounter = x} + \ \star\ 
\totext{0}
\Sleader\tuple{\Scounter = x-1, \Snone, \Sheadspresent=\Strue} +\ \star, &\text{for $x\ge 1.$}
\end{align}

%\begin{equation}
%\label{FE:zeroB}
%\Sleader\tuple{\Slevel = x, \Scount=c} + \ \star\ 
%\totext{0}
%\Sleader\tuple{\Slevel = x-1, \Scount=2, \Snone, \Sheadspresent=\Strue} +\ \star\ \quad\quad\text{if $x>0$}
%\end{equation}

When $x=1$, at the end of the round we move to the third epoch.
Otherwise, in the first half of the round application of the coin from 
the current level $\gamma(x)$ is guaranteed whp, for all active leader candidates. For $x \not= 2\Phi+3$:
\begin{align}
\label{FE:earlyA}
&\Sleader\tuple{\Salive,\Scounter=x,\Snone} + \Scoin\tuple{\Slevel=y}
\ \totext{early}\ \Sleader\tuple{\Salive,\Scounter=x,\Sheads,\Sheadspresent=\Sfalse} + \Scoin\tuple{\Slevel=y},\\
\label{FE:earlyB}
\begin{split}
&\Sleader\tuple{\Salive,\Scounter=x,\Snone} + \Scoin\tuple{\Slevel=y}\ \totext{early}\ \Sleader\tuple{\Salive,\Scounter=x,\Stails} + \Scoin\tuple{\Slevel=y},\\
&\Sleader\tuple{\Salive,\Scounter=x,\Snone} + Y\ \totext{early}\ \Sleader\tuple{\Salive,\Scounter=x,\Stails} + Y,
\end{split}
\end{align}
when $\gamma(x)\leq y$ and $\gamma(x)>y,Y\neq\Scoin$ respectively.

In the second half of the round the broadcast (via one-way epidemic) 
informing about drawn $\Sheads$ is performed as follows
\begin{align}
\label{FE:broadcastA}
&\Sleader\tuple{\Salive,\Stails,\Sheadspresent=\Strue} + \Sleader\tuple{\Sheadspresent=\Sfalse}
\ \totext{late}\ 
\Sleader\tuple{\Spassive,\Stails,\Sheadspresent=\Sfalse} + \Sleader\tuple{\Sheadspresent=\Sfalse},\\
\label{FE:broadcastB}
&\Sleader\tuple{\Sheadspresent=\Strue} + \Sleader\tuple{\Sheadspresent=\Sfalse}
\ \totext{late}\ 
\Sleader\tuple{\Sheadspresent=\Sfalse} + \Sleader\tuple{\Sheadspresent=\Sfalse}.
\end{align}

The following lemmas guarantee the correctness of the second epoch whp.

\begin{lemma}
\label{FE:onetoss}
There exists a constant $c>0,$ s.t., for any $q<1$ when $N\geq c\log n/q$ 
agents toss an asymmetric coin resulting in heads with probability $q$,
the following holds:
\begin{enumerate}
\item none of the agents draws {\em heads} with a negligible probability, and
\item more than $2q\cdot N$ agents draw {\em heads} with a negligible probability.
\end{enumerate}
\end{lemma}
\begin{proof}
	The probability that all agents draw {\em tails} is at most $(1-q)^{c\log n/q}\leq e^{-c\log n}=n^{-c}.$
	The expected number of agents which draw {\em heads} is $q\cdot N\geq c\log n$.
	By Chernoff bound the probability that more than $2q\cdot N$ agents draw {\em heads}
	is smaller than $e^{-Nq}\leq n^{-c}$.
\end{proof}

\begin{lemma}
\label{lem:fastelimination}
Applying coin (from level) $\Phi$ four times and then coins $\Phi-1,\Phi-2,\ldots,\ell+1,\ell$
twice reduces the number of active leader candidates to at most $c\log n/q$,
where $q$ is the probability of tossing heads by coin $\ell\geq 1$. 
\end{lemma}
\begin{proof}
	Induction is on $\ell$.
	In the base case when $\ell=\Phi$ we apply coin $\Phi$ four times.
	By Lemma \ref{FJ:loglog-4} we have $q\geq n^{-0.23}$.
	So applying this coin four times gives reduction
	of the number of active leader candidates to
	at most $\max\{16n\cdot n^{-4\cdot 0.23},c\log n/q\}=c\log n/q$ whp.
	
	Now assume the thesis holds for level $\ell+1$ and we prove it for $\ell$.
	By inductive hypothesis and Lemma \ref{FJ:reductB1},
	after application of coin $\ell+1$ twice there are at most 
	$c\log n/q'\leq 20c\log n/9q^2$ active leader candidates
	($q'$ is the counterpart of $q$ at level $\ell+1$).
	Applying this coin twice gives further reduction
	of active leader candidates to
	at most $\max\{80c\log n/9,c\log n/q\}=c\log n/q$ whp.
\end{proof}

%\input{s6_fig.tex}

%%%%%%%%%%%%
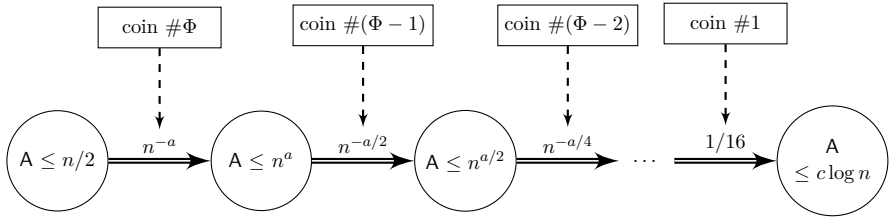
\begin{figure}[t]
\tikzstyle{line} = [draw, -latex', thick]
%\tikzstyle{empty} = [draw, text, text centered, node distance = 6em]
\tikzstyle{block} = [draw, rectangle, text centered, minimum width=6em, minimum height=2em]
\tikzstyle{circ} = [draw, circle, text centered, node distance = 6em, text width = 4em]

\centering
\begin{tikzpicture}[scale=0.7, transform shape]

\node[circ, anchor=center] (l0){\Salive $\le n/2$};
\node[circ, right=5em of l0] (l1){\Salive $\le n^{a}$};
\node[circ, right=5em of l1] (l2){\Salive $\le n^{a/2}$};
\node[text width = 2em, align = center, right=5em of l2] (mid){$\dots$};
\node[circ, right=5em of mid] (lphi){\Salive $\le c \log n$};
\path[line, double] (l0) -- node[anchor=south, name=arr0]{$n^{-a}$} (l1);
\path[line, double] (l1) -- node[anchor=south, name=arr1]{$n^{-a/2}$} (l2);
\path[line, double] (l2) -- node[anchor=south, name=arr2]{$n^{-a/4}$} (mid);
\path[line, double] (mid) -- node[anchor=south, name=arrphi]{$1/{16}$} (lphi);

\node[block, above=4em of arr0] (flip0){coin \#$\Phi$};
\node[block, above=4em of arr1] (flip1){coin \#$(\Phi-1)$};
\node[block, above=4em of arr2] (flip2){coin \#$(\Phi-2)$};
\node[block, above=4em of arrphi] (flipphi){coin \#$1$};

\path[line, thick, dashed] (flip0) -- (arr0);
\path[line, thick, dashed] (flip1) -- (arr1);
\path[line, thick, dashed] (flip2) -- (arr2);
\path[line, thick, dashed] (flipphi) -- (arrphi);

\end{tikzpicture}
\caption{An idealised scheme of the fast elimination process.}
\end{figure}

%%%%%%%%%%%%

%\todo[inline]{Uwagi: pierwszy przebieg pętli jest pusty, bo nie mamy nawet wszystkich leaderów i coinów.}

\section{Final elimination}
\label{sec:slow}

The protocol executes $\Theta(\log \log n)$ rounds of fast elimination 
(applying coins from level $\Phi$ down to $1$) concluding with $\bigo(\log n)$ 
active leaders left whp. All other leader candidates become passive.
The remaining task is to elect a single leader out of the remaining $\bigo(\log n)$ candidates.
In the final configuration there must be exactly one leader and $n-1$ followers and this situation should be maintained forever
(although the configuration can evolve indefinitely where the unique leader alternates between allowed leader states and followers move between dedicated follower states).
The main idea behind the solution is to iterate the following process.
Each of the remaining candidates picks some value at random, agents compute the maximum of these values in time $\bigo(\log n),$ and only the owners of the largest value remain leader candidates.
There can be at most $\bigo(\log\log n)$ values that one could pick from as we do not allow agents to operate on more states. In fact, we limit this choice to set $\{0,1\}$ as selecting random bits is easier, and drawing from a larger range does not necessarily provide us with a substantial gain in terms of time complexity.
Instead, we concentrate on the expected running time of unique leader election by observing that the number of bits to be drawn from $\{0,1\}$ (to conclude this process) is expected to be $\bigo(\log\log n)$.
This translates to the expected running time $\bigo(\log\log n\log n)$.
Note that to obtain a unique leader whp each agent has to pick $\bigo(\log n)$ bits, which translates to time $\bigo(\log^2 n)$.
This is why, instead, we focus on construction of a Las Vegas type algorithm which always elects exactly one leader in faster expected time $\bigo(\log\log n\log n)$. 

The most challenging problem to overcome during each iteration (when the next value from set $\{0,1\}$ is drawn) is to prevent all remaining leader candidates from becoming passive, which could happen due to phase clock desynchronisation with a negligible probability.
In fact a similar setback can also happen in the fast elimination epoch.
If being passive was equivalent to elimination we could accidentally cull all leaders, which is not permissible in a Las Vegas algorithm.
In order to prevent this from happening we continue using the leader modes
$\Salive$ (active), $\Spassive$ (passive) and we also introduce the $\Sdead$ (withdrawn) mode whose holders are followers.
The active and passive candidates may still become the unique leader, and we use a joint term {\em alive} candidates for these two groups.

In order to pick the value for an active leader candidate we utilise coin (from level)~$0$.
All alive candidates keep track of a counter $\Sslowness$ which is increasing
during the final elimination epoch.
This counter is ticking at rate which is slowing down in time.
The time elapsing between its $i$-th and $(i+1)$-st incrementation is on average
$\Theta(4^i\log n)$.
Only active candidates can increment this counter.
If a passive candidate detects an increment of the counter value (wrt to its own), it transitions into withdrawn state.
This process is now safe because withdrawing agents have enough evidence that
there are still active candidates with a higher $\Sslowness$ value.
In other words, we have a strong guarantee that all alive candidates do not transition into withdrawn state.
This is very important should the phase clock get unexpectedly desynchronised.
And indeed, the first increment of $\Sslowness$ counter makes
all candidates that became passive during the fast elimination withdrawn.
The counter $\Sslowness$ operates during the first $\Theta(n\log^2 n)$
interactions, when one leader is selected whp, so it has only $\Theta(\log \log n)$ states.
In slowing down this counter we rely on inhibitor agents ($\Sslow$) preprocessed 
at the same time as coins ($\Scoin$).
The counter $\Sslowness$ assures that if only one active candidate remains
in iteration $T$, then all other candidates become followers before
iteration $\bigo(T)$ whp.
Use of counter $\Sslowness$ is a new technique that makes it possible
to achieve expected stabilisation time $\bigo(\log\log n\log n)$ by withdrawing all
passive candidates soon after a single active leader remains whp.
The methods utilised in paper \cite{DBLP:conf/soda/GasieniecS18} do
not guarantee this effect and thus protocols from the latter have the
expected stabilisation time $\bigo(\log^2 n)$

\paragraph{Preprocessing}
begins with the first pass through 0 of the phase clock. Inhibitor agents keep track of $\Sslowness \in \{0,1, \ldots, \Psi\}$, $\Sslowmode \in \{\Sadvancing, \Sstopped\}$ (flag whether agent is \emph{advancing} or \emph{stopped}) and elevation $\Sactive \in \{\Slow, \Shigh\}$. 
The agents are initialized to $\Sslow\tuple{\Sslowness = 0, \Sadvancing, \Slow }$, and \Sslowness counts how many subsequent successful 
coin flips they managed to obtain.
\begin{align*}
\Sslow\tuple{\Sslowness = x, \Sadvancing} + Y \ &\totext{late}\ \Sslow\tuple{\Sslowness = x+1, \Sadvancing} +  Y,&\text{for } Y\not=\Scoin,\\
\Sslow\tuple{\Sslowness = x, \Sadvancing} + \Scoin \ &\totext{late}\ \Sslow\tuple{\Sslowness = x, \Sstopped} + \Scoin.
\end{align*}

We denote by $n_I$ the total number of inhibitor agents in $\Sslow$. By Lemma~\ref{lem:noninitialized}, $n_I = n/4 - \bigo(n/\log n)$.
Let $D_\ell$ be the number of agents that reach $\Sslowness$ $\ell$.

\begin{lemma}
\label{lem:d_bound}
After the first round of the clock $D_\ell = n4^{-\ell}(1\pm o(1))$ whp. 
\end{lemma}
\begin{proof}
	Let $D'_\ell = D_\ell + \ldots + D_\Psi$ be the number of inhibitor agents reaching slowness $\ell$ or higher and $p = \frac{n_c}{n}$ be the ratio of coins in the population.  By Lemma~\ref{lem:noninitialized} after $\bigo(n \log n)$ interactions of the first round $p = \frac{1}{4} - \bigo(1/\log n)$ and remains stable with high probability. An inhibitor agent reaches level $\ell$ by a series of $\ell$ successful synthetic coin flips, which happens with probability $p^\ell = 4^{-\ell}\left(1 - \ell \cdot \bigo(1 / \log n)\right)$. By Chernoff bound we have $D'_{\ell} = n_I \cdot p^{\ell} \pm \bigo(\sqrt{p^{\ell} \cdot n_I \log n_I})$, with high probability.  We have $\ell = \bigo(\log \log n)$, $n_I = \Theta(n)$ and $D_{\ell} = D'_{\ell} - D'_{\ell+1},$ for $\ell < \Psi$ and $D_{\Psi} = D'_{\Psi}.$ Thus there exists $D'_{\ell} = 4^{-\ell} \cdot n \cdot (1 \pm o(1))$ and the claimed bound holds. 
	
	In addition, we observe that  with high probability during the initial $\Theta(n \log n)$ interactions each inhibitor agent experiences $\Omega(\log n)$ interactions with 
	coin agents, determining its $\Sslowness$ during the second round of the clock.
\end{proof}

\paragraph
{Slowed-down inhibitor communication:}
Inhibitor agents get activated through interaction with the leader agents which reached the appropriate $\Sslowness$ value, and this communication is done via one-way epidemic (between inhibitors of the same \Sslowness), i.e.,
\begin{align}
\label{eq:firsttrigger}
&\Sslow\tuple{\Sslowness = x, \Sstopped, \Slow} + \Sleader\tuple{\Salive, \Sslowness = x} \to \Sslow\tuple{\Sslowness = x, \Sstopped, \Shigh} +  \Sleader\tuple{\Salive, \Sslowness = x},\\
\nonumber&\Sslow\tuple{\Sslowness = x, \star} + \Sslow\tuple{\Sslowness = x, \Shigh} \to \Sslow\tuple{\Sslowness = x, \Shigh} + \Sslow\tuple{\Sslowness = x, \Shigh}.
\end{align}

\paragraph{Safe withdrawal:}
All active leader candidates with $\Sslowness > 0$ are subject to coin-flipping rules \eqref{FE:earlyA} and \eqref{FE:earlyB}. More precisely, 
in the first half of the round each of them draws the coin from level $0$ whp.
As rules \eqref{FE:broadcastA} and \eqref{FE:broadcastB} apply to
agents with coin-flips resulting in success inform 
(via one-way epidemic) other agents accordingly.

We give below an updated reset rule (analogue of \eqref{FE:zeroA}) 
observing that this rule does not change the $\Sslowness$ value, as well as updated rules for leaders:
\begin{align}
\nonumber&\Sleader\tuple{\star,\Sheadspresent=\star} + \star \totext{0}\ \Sleader\tuple{\Snone, \Sheadspresent=\Strue} + \star,\\
\label{eq:killing_slow}&\Sleader\tuple{\star, \Sslowness = x} + \Sleader\tuple{\Sslowness = y}
\to 
\Sleader\tuple{\Sdead, \Sslowness = y} + \Sleader\tuple{\Sslowness = y},\quad\quad\quad\quad\quad\quad\quad\text{for } x<y,\\
\label{eq:first_passed}
&\Sleader\tuple{\Salive, \Sheads, \Sslowness= x} + \Sslow\tuple{\Sslowness = x,  \Shigh}
 \to \Sleader\tuple{\Salive, \Sheads, \Sslowness= x+1} + \Sslow\tuple{\Sslowness = x, \Shigh}.
\end{align}

Let $A \le c \log n$ be the number of active leaders with $\Sslowness = \ell$. 
Let $T_\ell$ be a random variable denoting the number of interactions between the first occurrence of an active leader candidate with $\Sslowness = \ell$ and the first occurrence of an active leader candidate with $\Sslowness = \ell+1$.

\begin{lemma}
\label{lem:sloweddownphases}
There exist constants $c_1,c_2>0$ such that for $\ell\leq\Psi$ we have
$\textrm{Pr}[T_\ell\leq c_1 4^\ell n \log n)] \le n^{-0.5}$ and  whp $T_\ell\leq c_2 4^\ell n \log n$.
\end{lemma}
\begin{proof}
	Consider the first interaction
	$t$ in which a leader candidate assumes $\Sslowness = \ell$.
	We are to prove inequalities on the number
	of interactions $T_\ell$ till
	the first interaction $t'$ in which a leader candidate assumes $\Sslowness = \ell+1$.
	
	We start with the first inequality, which is in fact a lower bound on $T_\ell$. 
	When the first leader with $\Sslowness = \ell$ occurs
	it starts propagation (via one-way epidemic) of state
	$\Shigh$ amongst inhibitors with $\Sslowness = \ell$.
	In the context of the lower bound argument, we can consider a situation in which $A \le c \log n$ informed agents spread rumour to $D_\ell$ uninformed agents in the population. By Lemma~\ref{lem:d_bound} $D_\ell = \Theta(n/4^\ell)$.
	We observe the following. If the informed part of population is of size $x$, then a single interaction increments this size with probability approximately $\frac{D_\ell}{n} \cdot \frac{x}{n} = 4^{-\ell} \cdot \frac{x}{n}(1\pm o(1))$. Thus, when $A < x < 1/2 \cdot D_\ell$ it takes $\Theta(n \cdot 4^\ell)$ interactions to go from $x$ informed agents to $2x$, with high probability (which follows from Chernoff bound). 
	So it takes $\Theta(4^\ell n\log n)$ interactions, i.e., more than $c_1 4^\ell n\log n$ for some $c_1>0$, to reach the sub-population of inhibitors that are $\Shigh$
	with $\Sslowness = \ell$ of cardinality $n^{0.4}$.
	During this time the probability of having an interaction
	incrementing value $\Sslowness$ to $\ell+1$
	is $\Theta(4^\ell n^{-0.6}\log^2 n)$.
	For $n$ large enough this probability is smaller than $n^{-0.5}$.
	
	With respect to the upper bound, consider the same communication process. Observe that if a single 
	inhibitor with $\Sslowness = \ell$ gets $\Shigh$, all inhibitors with $\Sslowness = \ell$ get $\Shigh$
	in $\bigo(4^\ell n \log n)$ subsequent interactions, with high probability. In further $\bigo(4^\ell n \log n)$ interactions 
	some active leader with $\Sslowness = \ell$ interacts
	with one of these inhibitors whp.
	Thus there is a constant $c_2>0$ such that $T_\ell\leq c_2 4^\ell n \log n$ whp.
\end{proof}

We now bound the time the protocol needs to elect a single leader whp. Recall that at the beginning of the last epoch, there are at most $c \log n$ active leaders, for a constant $c>0$.
\begin{lemma}
\label{lem:coinflips_to_single}
Assume that the preprocessing, the phase clock and coin propagations work properly.
After $\bigo(\log \log n)$ rounds in expectation and  $\bigo(\log n)$ rounds with high probability
the number of active leaders is reduced from $c \log n$ to $1$.
\end{lemma}
\begin{proof}
	Let sequence $F_0,F_1,\ldots$, $F_i$ contain the number of active leaders after $i$ rounds of elimination. Let $B$ be the number of rounds needed to obtain a single active leaders, that is $B = \min\{i : F_i = 1\}$. Let $F'_i$ be as follows: for $i \le B$, $F'_i = F_i$ and otherwise $F'_i = (5/6)^{i-B}$.
	Let $p=(1-\bigo(1/\log n))/4$ be the probability of drawing $\Sheads$ and $q=1-p$.
	
	Let $A_i$ be the event that all leader candidates drew $\Stails$ which happens with probability $q^{F_i}$.
	When $A_i$ occurs $F_{i+1} = F_i$.
	Otherwise, with probability $1-q^{F_i}$, the value of $F_{i+1}$ corresponds to the number of successes during consecutive $F_i$ coin-flips. 
	Thus,
	$\mathbb{E}[F_{i+1}\ |\ F_i] = F_i (p + q^{F_i})$. Assume that $F_i \ge 2$. Then $\mathbb{E}[F_{i+1}\ |\ F_i] \le F_i \cdot (13/16+o(1)) \le F_i \cdot 5/6$ for large enough $n$. Thus, by the definition of $F'_i$, we have $\mathbb{E}[F'_{i+1} | F'_i] \le F'_i \cdot 5/6$. 
	In turn $\Pr[B > i] = \Pr[F_i > 1] = \Pr[F'_i > 1] < \mathbb{E}[F'_i] = F_0 \cdot (5/6)^i$. Since $F_0 \le c \log n$, we get
	$$\text{Pr}[B > \log_{6/5}( c \log n \cdot n^\eta )] < n^{-\eta}, {\rm and}$$
	%and 
	$$\mathbb{E}[B] = \sum_{i=0}^{\infty} \Pr[B > i] \le \sum_{i=0}^{\infty} \min(1,F_0 \cdot (5/6)^i) = \bigo(\log F_0) + \bigo(1).\qedhere$$
\end{proof}

%%%%%%%%%%%%
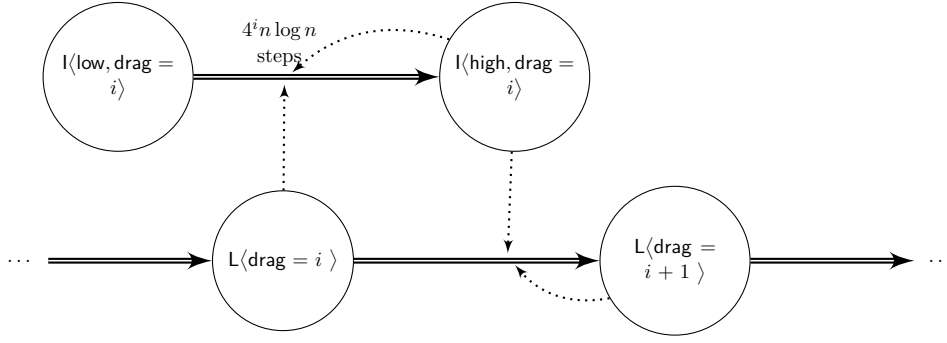
\begin{figure}[t]
\tikzstyle{line} = [draw, -latex', thick]
%\tikzstyle{empty} = [draw, text, text centered, node distance = 6em]
\tikzstyle{block} = [draw, rectangle, text centered, minimum width=6em, minimum height=2em]
\tikzstyle{circ} = [draw, circle, text centered, node distance = 8em, text width = 6em]

\centering
\begin{tikzpicture}[scale=0.7, transform shape]

\node[circ, anchor=center] (s0){$\Sleader\tuple{\Sslowness = i$}};
\node[circ, anchor=center, above left=4em and 3em of s0] (ilow){$\Sslow\tuple{\Slow, \Sslowness=i}$};
\node[circ, anchor=center, right=12em of ilow] (ihigh) {$\Sslow\tuple{\Shigh, \Sslowness=i}$};
\path[line, double] (ilow) -- (ihigh) node[anchor=south, pos = 0.35]{\makecell[c]{$4^i n \log n $\\steps}} node[anchor=center, name=lowhigh, pos=0.37]{};
\path[line, dotted] (s0) -- (lowhigh);
\node[circ, right=12em of s0] (s1){$\Sleader\tuple{\Sslowness = i+1$}};
\path[line, double] (s0) -- node[anchor=center, name=arr0, xshift=1.5em]{} (s1);
\path[line, dotted] (ihigh) -- (arr0);
\draw[line, dotted] (ihigh) to[bend right] (lowhigh);
\draw[line, dotted] (s1) to[bend left] (arr0);

\node[text width = 2em, align = center, right=8em of s1] (post){$\dots$};
\node[text width = 2em, align = center, left=8em of s0] (pre){$\dots$};
\path[line, double] (pre) -- (s0);
\path[line, double] (s1) -- (post);

\end{tikzpicture}
\caption{The implementation of slowing down $\Sslowness$ counter, where dotted arrows indicate enabled transitions.}
\end{figure}

%%%%%%%%%%%%

We now compute the time needed to elect a single leader, i.e., to preserve a single agent in state $\Salive$ and to change states of all agents in state $\Spassive$ to $\Sdead$.
\begin{lemma}
\label{lem:slow_elimination}
Assume that the preprocessing, the phase clock and coin propagations work properly
and that exactly one active leader enters $\Sslowness=\Psi$.
After time $\bigo(\log n \log \log n)$ in expectation and time $\bigo(\log^2 n)$ with high probability there is exactly one leader remaining in state $\Salive$ and all other candidates are moved to $\Sdead$.
\end{lemma}
\begin{proof}
	Let $T$ be the number of rounds it takes to go from $c \log n$ to $1$ of agents in state $\Salive$. By Lemma~\ref{lem:coinflips_to_single} $T$ is $\bigo(n \log n \log \log n)$ in expectation and  $\bigo(n \log^2 n)$ with high probability. It is enough to consider any round $T' > T$, s.t., between $T$ and $T'$ all leader agents increase their $\Sslowness$, since such interaction moves all $\Spassive$ to $\Sdead$.  
	Let $x$ be the highest $\Sslowness$ value achieved by
	a $\Spassive$ candidate.
	This candidate moves to state $\Sdead$ as soon as it
	encounters a higher value of $\Sslowness$ in another candidate.
	In the proof we use value 
	$T_x$ defined just before  
	Lemma~\ref{lem:sloweddownphases}.
	
	Note that by Lemma~\ref{lem:coinflips_to_single}
	$\Sslowness$\ value $x$ is smaller than $\Psi$.
	By Lemma~\ref{lem:sloweddownphases} we also have
	$\sum_{y=1}^\Psi T_y=\bigo(n\log^2 n)$ whp.
	Finally, the value of $\Psi$ is propagated amongst leaders
	in $\bigo(n\log n)$ interactions whp, which completes the proof of the time bound obtained whp.
	
	Let $T_A=\Theta(n\log n\log\log n)$ be the number
	of interactions of first two epochs.
	In order to obtain the improved time bound in expectation we observe that
	by Lemma~\ref{lem:sloweddownphases}
	there exists an integer constant $k$ such that for
	any $y$: $T_y+T_{y+1}+\cdots+T_{y+k}\geq c_1 4^y n\log n$
	whp, where $c_1$ is the constant defined in Lemma~\ref{lem:sloweddownphases}.
	Because of this $T\geq T_A+T_1+\cdots+T_x=\Omega(n\log n(\log\log n+4^x))$ whp.
	Note that $T_{x+1}\leq c_2 4^x n\log n$ whp,
	and the time of propagation of value $x+1$ amongst
	leaders is $\bigo(n\log n)$.
	Thus $T'=\bigo(n\log n(\log\log n+4^x))$.
	Finally, whp the extra time cost of getting all passive agents withdrawn increases the total number
	of interactions at most a constant number of times. With remaining negligible probability the expected value of $T'$ is at most the average number of interactions to get 
	all leader candidates to $\Sslowness=\Psi$
	which is $\bigo(n\log^2 n)$.
\end{proof}

\section{Slow backup protocol}

We have shown earlier that with high probability
our protocol behaves according to the scheme described in previous sections.
In this case the protocol elects a single leader in expected time $\bigo(\log n \log \log n)$
and with high probability in time $\bigo(\log^2 n)$. 
However, we still need to provide a full guarantee that new $\bigo(\log n \log \log n)$-time protocol elects a unique leader, 
even if the phase clock gets desynchronised at some point. 

And indeed, the successful conclusion is guaranteed by running simultaneously (in the background) the slow elimination protocol from~\cite{DBLP:conf/podc/AngluinADFP04} in which an encounter of two leader candidates results in transition of exactly one of them to a withdrawn (follower) state.
This must be done without disrupting neither of the three epochs of the fast leader election protocol presented above.
We achieve this by introducing a {\em seniority order} on alive leader candidates to break ties during direct encounters when always more senior leaders survive.
More precisely, we say that $\Sleader\tuple{\Salive}$ and $\Sleader\tuple{\Spassive}$ agents are mapped to the leader in the output, and $\Sleader\tuple{\Sdead}$, $\Scoin$, $\Sslow$, $X$, $\mathsf{D}$ and $0$ states are mapped to non-leaders.
We also adopt an extra interaction rule
in which when two agents $A,B \in\{\Sleader\tuple{\Salive},\Sleader\tuple{\Spassive}\}$ meet, $B$ changes its state to $\Sleader\tuple{\Sdead}$ when $A$ is more senior to $B,$ i.e., 
\begin{align}
\label{eq:slow_elim}
A + B\ \to\ A + \Sleader\tuple{\Sdead}.
\end{align}
In addition, the seniority order gives preference to agents with higher $\Sslowness,$ and if tied $\Sleader\tuple{\Salive}$ beats $\Sleader\tuple{\Spassive}$. Finally, the agent with a smaller $\Slevel$ wins, and $\Sheads$ wins with $\Snone$ and $\Stails$.

We first observe that with high probability rule \eqref{eq:slow_elim} may only speed up the elimination process analysed in Sections~\ref{sec:fast} and~\ref{sec:slow}, since it reduces the number of $\Sleader\tuple{\Salive}$ agents, and whp this rule never eliminates during one round \emph{all} agents with $\Sheads$.

\begin{lemma}
\label{lem:always_correct}
Throughout the execution of leader election protocol there is always at least one agent in state $\Sleader\tuple{\Salive}$ or $\Sleader\tuple{\Spassive}$.
\end{lemma}
\begin{proof}
Leader candidates equipped in state $\Sleader\tuple{\Salive}$ are formed by application of rule \eqref{eq:0_elimination}. Only rules \eqref{eq:killing_slow} and \eqref{eq:slow_elim} can change states of agents from $\Sleader\tuple{\Salive}$ and $\Sleader\tuple{\Spassive}$ to $\Sleader\tuple{\Sdead}$. 
However, neither of these rules can eliminate the last agent of this type which possesses the highest value of $\Sslowness$.
\end{proof}

We now arrive in the main result of this paper.

\begin{theorem}[{\bf Main result}]
\label{th:mainresult}
The leader election protocols presented in this paper always elects a unique leader. The election process concludes in the expected time $\bigo(\log n \log \log n)$, and time $\bigo(\log^2 n)$ with high probability.
\end{theorem}

\begin{proof}
By Lemma~\ref{lem:coins_speed} in $\bigo(\log n)$ rounds whp we elect a junta of the appropriate size as indicated by Lemma~\ref{FJ:loglog-4}. This junta starts the phase clock. By Lemma~\ref{lem:fastelimination}, fast elimination epoch leaves $\bigo(\log n)$ active leaders in $\bigo(\log n \cdot \log \log n)$ parallel time, with high probability. By Lemma~\ref{lem:slow_elimination}, slow elimination epoch leaves a single leader in expected parallel time $\bigo(\log n \cdot \log \log n)$ and in parallel time $\bigo(\log^2 n)$ with high probability. In addition, by Lemma~\ref{lem:always_correct} we never eliminate all leader candidates, and rule \eqref{eq:slow_elim} guarantees that whp in $\bigo(n)$ rounds in expectation and in $\bigo(n\log n)$ whp 
a single leader is chosen, which does not affect the overall running time.
\end{proof}

\section{Conclusion}

In this paper we presented the first $o(\log^2 n)$-time leader election protocol.
Our algorithm operates in parallel time $\bigo(\log n\log\log n)$
which is equivalent to $\bigo(n \log n\log\log n)$ pairwise interactions. 
The solution is always correct, however the obtained speed up 
%from $\bigo(\log^2 n)$ whp in~ to $\bigo(\log n\log\log n)$ 
refers to the expected time, and
our protocol works whp only in parallel time $\bigo(\log^2 n)$, as in \cite{DBLP:conf/soda/GasieniecS18}.
The first two epochs operate in time $\bigo(\log n\log\log n)$ whp. 
Thus the main bottleneck is the last epoch when we reduce the number of 
leader candidates from $\bigo(\log n)$ to a single one. 
We would like to claim that likely the hardest problem in leader election is reduction from two leader candidates to a single one. And if one is able to solve this problem rapidly whp, they should be able to solve leader election with the same time complexity too.

\clearpage

\bibliographystyle{alpha}
\bibliography{population}
\end{document}